\documentclass[12pt,reqno]{amsart}
\usepackage{dsfont, amssymb,amsmath,amscd,latexsym, amsthm, amsxtra,amsfonts}
\usepackage[active]{srcltx}
\usepackage{stmaryrd}
\usepackage{amsmath}
\usepackage{amssymb}
\usepackage{txfonts}
\usepackage{dsfont, amssymb,amsmath,amscd,latexsym, amsthm, amsxtra,amsfonts}
\usepackage{graphicx}
\usepackage{float}
\newtheorem{Them}{Theorem}[section]
\newtheorem{Prop}{Proposition}[section]
\newtheorem{Remark}{Remark}[section]

\newtheorem{CR}{Corollary}[section]
\newtheorem{EX}{Example}[section]

\begin{document}
\makeatletter
\def\@setauthors{%
\begingroup
\def\thanks{\protect\thanks@warning}%
\trivlist \centering\footnotesize \@topsep30\p@\relax
\advance\@topsep by -\baselineskip
\item\relax
\author@andify\authors
\def\\{\protect\linebreak}%
{\authors}%
\ifx\@empty\contribs \else ,\penalty-3 \space \@setcontribs
\@closetoccontribs \fi
\endtrivlist
\endgroup }
\makeatother
 \baselineskip 17pt
\title[{{\bf stock loan and capped Stock loan}}]
 {{  Variational
 inequality method in stock loans }}
\author[{\bf Zongxia Liang and Weiming Wu } ]
{ Zongxia Liang \\ Department of Mathematical Sciences, Tsinghua
University, \\ Beijing 100084, China \quad  Email:
zliang@math.tsinghua.edu.cn\\
Weiming Wu\\ Department of Mathematical Sciences, Tsinghua
University,\\ Beijing 100084, China  \quad  Email:
wuweiming03@gmail.com } \maketitle
\begin{abstract}
 In this paper we first introduce two new financial products:
  stock loan and capped stock loan. Then we  develop a pure variational
 inequality method to establish explicitly the values of these stock loans.
 Finally, we work out  ranges of fair values of parameters associated with
 the loans.
 \end{abstract}
 \vskip 10pt \noindent
\setcounter{equation}{0}
\section{{\small {\bf Introduction}}}
 \vskip 10pt\noindent
Stock loan is a simple economy where a client(borrower), who owns
one share of a stock, borrows a loan  of amount $q$ from a bank
(lender) with one share of stock as collateral. The bank charges
amount $c(0\leq c\leq q)$ from the client for the service.  The
client may regain the stock by repaying  principal and interest
(that is, $qe^{\gamma t}$, here $\gamma$ is continuously compounding
{\sl loan} interest rate ) to the bank, or surrender the stock
instead of repaying the loan at any time. It is a currently popular
financial product. It can create liquidity while overcoming the
barrier of large block sales, such as triggering tax events or
control restrictions on sales of stocks. It also can  serve as a
hedge against a market down: if the stock price goes down, the
client may forfeit the loan at  initial time; if the stock price
goes up, the client keeps all the upside by
repaying the principal
and interest. Therefore, the stock loan has unlimited liability. As
a result, stock loan transforms  risk into the bank.
 To reduce the risk, the bank introduces cap $L$ feature
in the stock loan because the cap adds a further incentive to
exercise early. Such a loan is called a {\sl capped stock loan}
throughout this paper. Since the capped stock loan has limited
liability and lower risk, it therefore will be an attractive
instrument to market for an issuer, or  to hold short for an
investor like American call option with cap in financial market(cf.
Broadie and Detemple \cite{BRDE}(1995)). \vskip 10pt \noindent
 The wide acceptance and  popular feature of the capped
 (uncapped)
stock loan in the marketplace, however, greatly depend on how to
make successfully these kinds of financial  products.  More
precisely, how to work out right values of the parameters $(q,
\gamma, c, L)$ is  a natural and key problem in negotiation between
the client and the bank at initial time. Unfortunately, to the
authors' best knowledge, it seems that few results on the topic have
been reported in existing literature. The main goal of the present
paper is to develop a pure variational inequality method to solve
this kind of  problems. \vskip 10pt \noindent
 We explain   major difficulty and main idea of
solving  the problem  as follows. We  formulate the capped
(uncapped) stock loan as a  perpetual American option with negative
interest rate. We denote by $f(x)$ the initial value  of this
option. The problem  can be reduced to calculating  the function
$f(x)$ for determining the ranges of fair values of the  parameters
$(q, \gamma, c, L)$. According  to the conventional variational
inequality method, the $f(x)$ must satisfy a variational inequality,
we calculate  $f(x)$ by initial condition $f(0)=0$ and smooth-fit
principle. However, because of negative interest rate the initial
condition $f(0)=0$ does not work, i.e., the conventional variational
inequality method can not solve this option with negative interest
rate. Moreover, the presence of the cap also complicates the
valuation procedure(if we focus on studying the capped stock loan).
So we need to develop the variational inequality method to deal with
the case of negative interest rate. As payoff process of the capped
(uncapped) stock loan is a Markov process, the optimal stopping
 time must be a hitting time( Remark \ref{R31} below) from which we guess that
 $f'(0)=0$. In addition, we observe that the condition do work  in the case
 of negative interest rate but it has no use to dealing  with the case of
 non-negative interest rate. Based on the  conjecture and observation, we first
establish explicitly the value $f(x)$ of the capped stock loan by a
new pure variational inequality method. Then we use  the expression
of $f(x)$ to work out the ranges of fair values of parameters
associated with
 the  loan.  Finally, as a special case of our main result, we also get the
same conclusion on uncapped stock loan as in Xia and Zhou
\cite{stock} proved by a pure probability approach.
\vskip10pt\noindent
 The paper is organized as follows: In Section
2, we formulate a mathematical model of capped stock loan and it is
considered as a perpetual American option with a possibly negative
interest rate. In Section 3, we extend standard variational
inequality method to  the case of negative interest rate and
calculate the initial value and  the value process of the
capped(uncapped) stock loan. In Section 4, we work out the ranges of
fair values of parameters associated with the loan based on
 the results in previous sections. In Section 5, we
present two examples to explain how the cap impacts on the initial
value of uncapped stock loan.
\vskip 10pt\noindent
\setcounter{equation}{0}
\section{{\small {\bf Mathematical model}}}
\vskip 10pt \noindent In this section the standard Black-Scholes
model in a continuous financial market consists of two assets: a
risky asset stock $S$ and a risk-less bond $B$. The uncertainty  is
described by a standard Brownian motion $\{ \mathbf{W}_{t}, t\geq
0\}$ on a risk-neutral complete probability space $(\Omega, \mathcal
{F}, \{ \mathcal {F}_{t}\}_{t\geq 0}, P)$, where $\{ \mathcal
{F}_{t}\}_{t\geq 0}  $ is the filtration generated by $ \mathbf{W}
$, $\mathcal {F}_{0}=\sigma\{ \Omega, \emptyset\}$ and $ \mathcal
{F}=\sigma\{ \bigcup_{t\geq 0}\mathcal {F}_{t} \}$. The risk-less
bond $B$ evolves according to the following  dynamic system,
\begin{eqnarray*}
dB_{t}=rB_{t}dt,\   \  r>0,
\end{eqnarray*}
where $ r$ is  continuously compounding interest rate. The stock
price $S$ follows a geometric Brownian motion,
\begin{eqnarray}  \label{E2.1}
S_{t}=S_{0}e^{(r-\delta-\frac{\sigma^{2}}{2})t+\sigma\mathbf{W}_{t}},
\end{eqnarray}
where $S_{0}$ is  initial stock price, $\delta\geq0$ is  dividend
yield and $\sigma >0$ is  volatility. The discounted payoff process
of the capped stock loan is defined by
\begin{eqnarray*}
Y(t)=e^{-rt}(S_{t}\wedge Le^{\gamma t}-qe^{\gamma t})_{+}.
\end{eqnarray*}
Since $Y(t)\geq 0, \  a.s.$  and $Y(t)>0$ with a positive
probability, to avoid arbitrage,  throughout this paper we assume
that
\begin{eqnarray}
 S_{0}-q+c>0 .
\end{eqnarray}
According to theory of American contingent claim, the initial value
of this capped stock loan is
\begin{eqnarray}\label{reward2}
f(x)&=&\sup\limits_{\tau \in \mathcal {T}_{0}}{\bf E}\big
[e^{-r\tau}(S_{\tau}\wedge Le^{\gamma \tau}-qe^{\gamma
\tau})_{+}\big ]\nonumber\\
&=&\sup\limits_{\tau \in \mathcal {T}_{0}}{\bf E}\big [e^{-\tilde
{r}\tau}(\tilde{S}_{\tau}\wedge L-q)_{+}\big ],
\end{eqnarray}
where $\tilde{r}=r-\gamma$, $\tilde{S}_{t}=e^{-\gamma
t}S_{t},\tilde{S}_{0}=S_{0}=x$, $ L>S_{0}$ and $\mathcal {T}_{0}$
denotes all $\{ \mathcal {F}_{t}\}_{t\geq 0}$-stopping times.
 The value process of this capped stock loan is
 \begin{eqnarray}\label{C24}
V_{t}=\sup\limits_{\tau \in \mathcal {T}_{t}}{\bf E}\big[e^{-
{r}(\tau-t)}(S_{\tau}\wedge Le^{\gamma \tau}-qe^{\gamma
\tau})_{+}|\mathcal{F}_{t}\big],
\end{eqnarray}
i.e.,
\begin{eqnarray*}
e^{-rt}V_{t}=\sup\limits_{\tau \in \mathcal {T}_{t}}{\bf E}\big
[e^{-\tilde {r}\tau}(\tilde{S}_{\tau}\wedge
L-q)_{+}|\mathcal{F}_{t}\big ],
\end{eqnarray*}
where $\mathcal {T}_{t}$ denotes all $\{ \mathcal {F}_{t}\}_{t\geq
0}$-stopping times $\tau $ with $\tau \geq t$ a.s.. Since the fair
values of $q, \gamma,c, L$ should be such that $f(S_0)=S_0-q+c$, the
range of the fair values of the parameters $( q, \gamma, c, L)$
reduce to calculating the $f(S_0)$. Because $\tilde{r}=r-\gamma\leq
0$, the problem is essentially to calculate the initial value  of
 a conventional perpetual American call option with a
possibly negative interest rate. We have the following.
\begin{Prop}\label{inequality}
$(x-q)_{+}\leq f(x)\leq x$ for  $x\geq 0$. $f(x)$ is continuous and
nondecreasing on $(0,\infty)$
\end{Prop}
\begin{proof}
 Using the same way as in  Xia and Zhou \cite{stock}(2007), we
have $(x\wedge L-q)_{+}\leq f(x)\leq x\wedge L$ for  $x\geq0$. The
$f(x)$ is continuous and nondecreasing on $(0,\infty)$ can be proved
by  the optional sampling theorem.
\end{proof}
\vskip 10pt \noindent
 \setcounter{equation}{0}
\section{{\small {\bf Variational
 inequality method in stock loans}}}
\vskip 10pt \noindent
 In this section we develop a pure variational
 inequality method in the case of
negative interest rate to establish explicitly the  value $ f(x)$
  of the capped (uncapped) stock loan.  The key point is to replace
the initial condition $f(0)=0$ in the conventional case with a new
one $f'(0+)=0$. The detailed observation will be given in {\bf
Remark} \ref{R31} below. We find that the condition $f'(0+)=0$ holds
for the conventional perpetual American call option but no use in
determining free constants. Now we star with the following.
\begin{Prop} \label{Prop1}
Assume that $\delta>0$ and $\gamma-r+\delta\geq 0$ or $\delta=0$ and
$\gamma-r>\frac{\sigma^{2}}{2}$. Let $\lambda_1 $, $\lambda_2$ and $\mu  $ be defined by
\begin{eqnarray}\label{e35}
\left\{
\begin{array}{l l l}
\lambda_{1}=\frac{-\mu+\sqrt{\mu^{2}-2(\gamma-r)}} {\sigma},\
\lambda_{2}=\frac{-\mu-\sqrt{\mu^{2}-2(\gamma-r)}}{\sigma},\\
\mu=-(\frac{\sigma}{2}+\frac{\gamma-r+\delta}{\sigma}).
\end{array}
\right.
\end{eqnarray}
\vskip 10pt\noindent
(i)If $L\geq b^{}$,  $h(x)\in \mathcal
{C}\big([0,\infty)\big )\cap\mathcal
{C}^{1}\big((0,\infty)\setminus \{L\}\big)\cap \mathcal
{C}^{2}\big((0,\infty)\setminus \{b^{},L\}\big )$ and solves the
 variational inequality
\begin{eqnarray}\label{e31}
\left\{
\begin{array}{l l l l l}
\frac{1}{2}\sigma^{2}x^{2}h^{''}+(\tilde{r}-\delta)xh^{'}
-\tilde{r}h=0,  & x\in [0,b)\cup (L, \infty),\\
h(x)=x-q,\quad &x\in [b, L],\\
h(x)>(x-q)_{+},& x<b,\\
h(x)\leq x, & x \geq 0,\\
h'(0+)=0,\ h(b)=b-q,h^{'}(b)=1,
\end{array}
\right.
\end{eqnarray}
 then
\begin{eqnarray}\label{solution21}
h(x)=\left\{
\begin{array}{l l l}
(b^{}-q)(\frac{x}{b^{}})^{\lambda_{1}},&0< x\leq b^{},\\
x-q ,& b^{}<x<L,\\
(L-q)(\frac{x}{L})^{\lambda_{2}},&x\geq L.
\end{array}
\right.
\end{eqnarray}
Moreover,
\begin{eqnarray}\label{b}
b=\frac{q\lambda_1}{\lambda_1-1}.
\end{eqnarray}
 (ii) If $L< b^{}$,  $h(x)\in \mathcal
{C}\big([0,\infty)\big )\cap\mathcal
{C}^{1}\big((0,\infty)\setminus \{L\}\big)\cap \mathcal
{C}^{2}\big((0,\infty)\setminus \{L\}\big )$ and solves the
 variational inequality
\begin{eqnarray}\label{equivalent22}
\left\{
\begin{array}{l l l l}
\frac{1}{2}\sigma^{2}x^{2}h^{''}+(\tilde{r}-\delta)xh^{'}-\tilde{r}h=0,
& x \in [0,L) \cup (L, \infty ),\\
h(x)>(x-q)_{+}, &  x<L,\\
h(x)\leq x, & x \geq 0,\\
h'(0+)=0, h(L)=L-q,
\end{array}
\right.
\end{eqnarray}
then
\begin{eqnarray}\label{solution22}
h(x)=\left\{
\begin{array}{l l l}
(L-q)(\frac{x}{L})^{\lambda_{1}},&0<x<L,\\
L-q ,& x=L,\\
(L-q)(\frac{x}{L})^{\lambda_{2}},&x> L.
\end{array}
\right.
\end{eqnarray}
\end{Prop}
\begin{proof} We only deal with the part (i). The
 part (ii) can be treated similarly.
 Solving the first second-order differential
  equation of the(\ref{e31}), we get
\begin{eqnarray}\label{e39}
h(x)=\left\{
\begin{array}{l l l}
C_1x^{\lambda_1}+ C_2x^{\lambda_2}, & x \in [0, b),\\
x-q, & x\in [b, L],\\
C_3x^{\lambda_1}+ C_4x^{\lambda_2},& x \in (L, \infty ),
\end{array}
 \right.
\end{eqnarray}
where $C_i$, $i=1,2,3,4$, and $b$ are free constants  to be
determined, $ \lambda_1$ and $\lambda_2 $ are defined by
(\ref{e35}). \vskip 5pt \noindent Note that { \sl if $\delta >0$
then $\lambda_1 >1 > \lambda_2
>0 $, and if $\delta=0 $ and $\gamma-r> \frac{\sigma^2}{2}  $ then
$\lambda_1=\frac{2(\gamma-r  )}{\sigma^2}>1= \lambda_2  $}. Using
$h'(0+)=0$ and  $h(x)<x$, we have $C_2= C_3=0 $. Substituting
$C_2=C_3=0$ into (\ref{e39}) and applying $h(b)=b-q$, $ h'(b)=1 $ as
well as the principle of smooth fit
  in the resulting expression of $h(x)$ yields that
$$b-q=C_1b^{\lambda_1}, \quad   L-q= C_4 L^{\lambda_2},\quad  \lambda_1C_1
b^{\lambda_1-1}=1.$$ \vskip 5pt \noindent Solving this system for
$b$, $ C_1$ and $C_4$, we get $b=\frac{q\lambda_1}{\lambda_1-1}$, $
C_1=\frac{b-q}{b^{\lambda_1} }$ and  $
C_4=\frac{L-q}{L^{\lambda_2}}$.  Hence the equations
(\ref{solution21}) and (\ref{b}) follow. Since  $\lambda_1 >1 \geq
\lambda_2 >0 $ , the function $h(x)$ defined by (\ref{solution21})
obviously belongs to $\mathcal {C}\big([0,\infty)\big )\cap\mathcal
{C}^{1}\big((0,\infty)\setminus \{L\}\big)\cap \mathcal
{C}^{2}\big((0,\infty)\setminus \{b^{},L\}\big ).$ Thus we complete
the proof.
\end{proof}
\vskip 2pt \noindent
\begin{Remark}\label{R31}
From (\ref{e35}) we know that if $\widetilde{r}=r-\gamma>0 $( the
conventional case ) then $\lambda_1\geq 1>0$ and $\lambda_2 <0 $.
Consequently, we see from (\ref{e39}) that the function
$h(x)=C_2x^{\lambda_2}$ is rejected by the initial condition
$h(0)=0$, i.e., $C_2=0$. So by the same way as in proof of
Proposition \ref{Prop1} above, we can determine other free constants
$C_1$, $C_3$, $ C_4$ and $b$. \vskip 5pt \noindent If
$\widetilde{r}=r-\gamma\leq 0 $( the present case ) then $\lambda_1>
1\geq \lambda_2 >0  $. As opposed to the conventional case, we can
not deduce from $h(0)=0$ that $C_2=0$. This is the major difficulty
appearing in  variational
 inequality method in the case of stock loan. However, we note that
 in present case $h'(x)=\lambda_1C_1x^{\widetilde{\lambda_1}} +
 \lambda_2C_2x^{\widetilde{\lambda_2}}$ with
 $\widetilde{\lambda_1}=\lambda_1-1 > 0 $ and
 $\widetilde{\lambda_2}=\lambda_2-1 < 0 $. So if $h'(0+)=0$, the
 present case can be reduced to the conventional case. On the
 other hand, Markov property implies that the optimal stopping
 time must be a hitting time. In view  of the fact, we conjecture
that $h'(0+)=0$ is correct. The proof of the conjecture will be
given in Proposition \ref{Prop2} below. Comparing with  the
conventional case,  the $h'(0+)=0$ will play an important role in
developing  pure variational
 inequality method in the case of
stock loan.
\end{Remark}
\vskip 2pt \noindent
\begin{Prop}\label{Prop2}
Assume that  $\delta=0$, $\gamma-r>\frac{\sigma^{2}}{2}$ and  $
g(x)={\bf E}\big [e^{-\tilde{r}(\tau\wedge \varsigma )}
(\tilde{S}_{\tau\wedge \varsigma   }-q)_{+}I_{\{\varsigma <\infty\}}
\big ] $  for  $b\geq q$,  $ L\geq 0$, stoping times $\tau $ and
$\varsigma$.   Then $g'(0+)=0$.
\end{Prop}
\begin{proof} Using $g(0)=0$, we have for $x>0$
\begin{eqnarray}\label{e37}
0&\leq &\frac{g(x)-g(0)}{x-0}\nonumber\\
&=& {\bf E}\large\{ \large (\exp\{-\frac{1}{2}
\sigma^2\tau\wedge \varsigma+\sigma\mathbf{W}_{\tau\wedge \varsigma
} \}\nonumber\\
&& \qquad -\frac{q}{x}\exp\{ (\gamma-r )\tau\wedge \varsigma\}
\large )_+I_{\{    \varsigma< +\infty\}}
\large \}\nonumber\\
&\leq &  {\bf E}\large\{ \large (\exp\{-\frac{1}{2}
\sigma^2\tau\wedge \varsigma+\sigma\mathbf{W}_{\tau\wedge \varsigma
} \}-\frac{q}{x} \large )_+I_{\{ \varsigma< +\infty\}} \large \}.
\end{eqnarray}
Note that $\large\{\exp\{-\frac{1}{2}\sigma^2t+\sigma \mathbf{W}_t
\}, t\geq 0 \large\} $ is a strong martingale w.r.t. $\{ \mathcal
{F}_{t}\}_{t\geq 0}  $. It follows that
\begin{eqnarray}   \label{C1}
 {\bf E}\large\{ \exp\{-\frac{1}{2}
\sigma^2\tau\wedge \varsigma+\sigma\mathbf{W}_{\tau\wedge \varsigma
}\} \large\}=1.
\end{eqnarray}
 To let $x\rightarrow 0$ in (\ref{e37}), we need to dominate the
 right-hand side. Because
 \begin{eqnarray*} \label{C2}
  \Large (\exp\{-\frac{1}{2}
\sigma^2\tau\wedge \varsigma+\sigma\mathbf{W}_{\tau\wedge \varsigma
} \}-\frac{q}{x} \Large )_+I_{\{ \varsigma< +\infty\}}\leq
\exp\{-\frac{1}{2}
\sigma^2\tau\wedge \varsigma+\sigma\mathbf{W}{\tau\wedge \varsigma } \},\nonumber\\
\end{eqnarray*}
\begin{eqnarray*}\label{C3}
 \lim_{x\rightarrow 0+}\Large (\exp\{-\frac{1}{2}
\sigma^2\tau\wedge \varsigma+\sigma\mathbf{W}_{\tau\wedge \varsigma
} \}-\frac{q}{x} \Large )_+I_{\{ \varsigma< +\infty\}}=0. \quad   a.s.,
 \end{eqnarray*}
and (\ref{C1}), the dominated convergence theorem allows us to let
$x\rightarrow 0 $ in (\ref{e37}) and obtain
 \begin{eqnarray*}
 0\leq \lim_{x\rightarrow 0+}\frac{g(x)-g(0)}{x-0}
 \leq \lim_{x\rightarrow 0+} {\bf E}\large\{ \big (\exp\{-\frac{1}{2}
\sigma^2\tau\wedge \varsigma+\sigma\mathbf{W}_{\tau\wedge \varsigma
} \}-\frac{q}{x} \big )_+I_{\{    \varsigma< +\infty\}} \large \}=0.
 \end{eqnarray*}
Thus $g^{'}(0+)=0$.
\vskip 5pt\noindent
\end{proof}
\vskip 5pt\noindent In view of Proposition \ref{Prop1} and
variational
 inequality method, $\tau_b$ must be the optimal time, where $b$ is
 defined by (\ref{b}). Now we show the fact.
 \vskip 5pt\noindent
\begin{Prop}\label{Prop3}
Let $b$  and $h(x) $ be defined by (\ref{b}), (\ref{solution21}) and
(\ref{solution22}   ) respectively,  $ E(x)\equiv {\bf E}\big
[e^{-\tilde{r}(\tau_{b}\wedge\tau_{L})}
(\tilde{S}_{\tau_{b}\wedge\tau^{}_{L}}\wedge
L-q)_{+}I_{\{\tau^{}_{L}<\infty\}} \big ]$ for $L\geq 0 $, where
$x=S_0$, $\tau_{b^{}}=\inf{\{t\geq 0:e^{-\gamma t}S_{t}\geq b^{}\}}$
and $\tau^{}_{L}=\inf{\{t\geq 0:e^{-\gamma t}S_{t}=L\}}$. Then
$E(x)=h(x)$.
\end{Prop}
\begin{proof}
Using formula 2.20.3 in Section 9, Part II of \cite{Handbook}(2000),
 we have
 \begin{eqnarray}\label{C5}
{\bf E}\big
[e^{-\tilde{r}(\tau_{b}\wedge\tau_{L})}
I_{\{\tau_{b}\wedge\tau_{L}<\infty\}}\big
]=\left\{
\begin{array}{l l}
(\frac{x}{b})^{\lambda_{1}},x<b,\\
(\frac{x}{L})^{\lambda_{2}},x\geq L.
\end{array}
\right.
\end{eqnarray}
Based on the (\ref{C5}), the proof of Proposition \ref{Prop3} will
be accomplished in two cases, namely, $L>b$ and $L\leq b$. \vskip
5pt\noindent {\bf Case of $L>b$. }  We shall distinguish three
subcases, i.e., $x<b$, $L\geq x\geq b$ and  $ x\geq L $. \vskip
5pt\noindent If $x<b$ then $\tau_{b}\leq\tau_{L}$. Using(\ref{C5}),
we calculate
\begin{eqnarray}\label{C6}
E(x)&=&{\bf E}\big
[e^{-\tilde{r}(\tau_{b}\wedge\tau_{L})}
(\tilde{S}_{\tau_{b}\wedge\tau_{L}}\wedge
L-q)_{+}I_{\{\tau_{L}<\infty\}}
\big ]\nonumber\\
&=&{\bf E}\big
[e^{-\tilde{r}(\tau_{b}\wedge\tau_{L})}
(\tilde{S}_{\tau_{b}}\wedge
L-q)_{+}I_{\{\tau_{L}<\infty\}}
\big ]
\nonumber\\
&=&(b-q)(\frac{x}{b})^{\lambda_{1}}.
\end{eqnarray}
If  $L\geq x\geq b$ then $\tau_{b}=0$ and so
\begin{eqnarray}\label{C7}
E(x)&=&{\bf E}\big [e^{-\tilde{r}(\tau_{b}\wedge\tau_{L})}
(\tilde{S}_{\tau_{b}\wedge\tau_{L}}\wedge
L-q)_{+}I_{\{\tau_{L}<\infty\}}
\big ]\nonumber\\
&=&{\bf E}\big [(\tilde{S}_{0}\wedge L-q)_{+}\big ]=
(x-q).
\end{eqnarray}
If $ x> L $ then $\tilde{S}_{\tau_{b}\wedge\tau_{L}}\geq L    $. By
using (\ref{C5}),
\begin{eqnarray}\label{C8}
E(x)&=&{\bf E}\big [e^{-\tilde{r}(\tau_{b}\wedge\tau_{L})}
(\tilde{S}_{\tau_{b}\wedge\tau_{L}}\wedge
L-q)_{+}I_{\{\tau_{L}<\infty\}}
\big ]\nonumber\\
&=&(L-q){\bf E}\big [e^{-\tilde{r}(\tau_{b}\wedge\tau_{L})}
I_{\{\tau_{L}<\infty\}}
\big ]\nonumber\\
&=&(L-q)(\frac{x}{L})^{\lambda_{2}}.
\end{eqnarray}
Comparing with (\ref{solution21}), the equations (\ref{C6}),
(\ref{C7}) and (\ref{C8}) yield that $ E(x)=h(x)$. \vskip
5pt\noindent {\bf Case of $b\geq L$.}  We shall distinguish two
subcases, i.e.,$x<L$ and $x\geq L   $. \vskip 5pt\noindent If $x<L$
then $\tau_b \geq \tau_l  $ . Using (\ref{C5}), we have
\begin{eqnarray}\label{C9}
 E(x)&=&{\bf E}\big
[e^{-\tilde{r}(\tau_{b}\wedge\tau_{L})}
(\tilde{S}_{\tau_{b}\wedge\tau{L}}\wedge
L-q)_{+}I_{\{\tau_{L}<\infty\}}
\big ]\nonumber\\
&=&{\bf E}\big
[e^{-\tilde{r}(\tau_{b}\wedge\tau_{L})}
(\tilde{S}_{\tau{L}}\wedge
L-q)_{+}I_{\{\tau_{L}<\infty\}}
\big ]\nonumber\\
&=&
(L-q)(\frac{x}{L})^{\lambda_{1}}.
\end{eqnarray}
If $ x\geq L $ then $\tilde{S}_{\tau_{b}\wedge\tau_{L}}\wedge
L=L     $. So by using (\ref{C5})
\begin{eqnarray}\label{C10}
E(x)&=&{\bf E}\big [e^{-\tilde{r}(\tau_{b}\wedge\tau_{L})}
(\tilde{S}_{\tau_{b}\wedge\tau_{L}}\wedge
L-q)_{+}I_{\{\tau_{L}<\infty\}}
\big ]\nonumber\\
&=&(L-q){\bf E}\big [e^{-\tilde{r}(\tau_{b}\wedge\tau_{L})}
I_{\{\tau_{L}<\infty\}}
\big ]\nonumber\\
&=&(L-q)(\frac{x}{L})^{\lambda_{2}} .
\end{eqnarray}
Comparing  the equations (\ref{C9}) and (\ref{C10}) with the (\ref{solution22}), we see that
$E(x)=h(x)$.
\end{proof}
\vskip 5pt\noindent We now return to main result of this section. It
states that the initial value of the capped stock loan is just
$f(x)$ defined by (\ref{solution21}) and (\ref{solution22})
respectively, where $x=S_0$. \vskip 5pt\noindent
\begin{Them}\label{main2}
Assume that $\delta>0$ and $\gamma-r+\delta\geq 0$ or $\delta=0$ and
$\gamma-r>\frac{\sigma^{2}}{2}$. Let  $h(x)$ be defined by
(\ref{solution21}) and (\ref{solution22})respectively,  and
 $f(x)$ be defined by (\ref{reward2}). Then $ f(x)=h(x)$ for $x\geq 0$, and
 $\tau_{b}\wedge\tau^{}_{L}$ is the optimal stopping time.
\end{Them}
\begin{proof} In view of Proposition \ref{Prop3},
 $$f(x)\equiv\sup\limits_{\tau \in \mathcal {T}_{0}}{\bf E}\big
[e^{-\tilde {r}\tau}(\tilde{S}_{\tau}\wedge
L-q)_{+}I_{\{\tau_{L}<\infty\} }\big ]\geq h(x).$$ Therefore we only need  to prove that for any stopping
time $\tau$
\begin{eqnarray} \label{Appendx1}
h(x)\geq {\bf E}\big [e^{-\tilde{r}(\tau\wedge\tau^{}_{L})}(\tilde
{S}_{\tau\wedge\tau^{}_{L}}\wedge
L-q)_{+}I_{\{\tau^{}_{L}<\infty\}}\big ].
\end{eqnarray}
From Proposition \ref{Prop1} and the expression of $h(x)$ as well as
  $\lambda_1 >1 \geq \lambda_2 >0 $, we know that
$h(x)\in \mathcal
{C}\big([0,\infty)\big )\cap\mathcal
{C}^{1}\big((0,\infty)\setminus \{L\}\big)\cap \mathcal
{C}^{2}\big((0,\infty)\setminus \{b^{},L\}\big )$  or
 $h(x)\in \mathcal
{C}\big([0,\infty)\big )\cap\mathcal {C}^{1}\big((0,\infty)\setminus
\{L\}\big)\cap \mathcal {C}^{2}\big((0,\infty)\setminus \{L\}\big )$
and $0< h''(b\pm)<+\infty$, $ 0 < h'(L\pm)< +\infty $ and $ 0 <
h''(L\pm)< +\infty $.  As a result,  the  generalized
  It\^{o}'s formula(cf.Karatzas and Shreve\cite{Brownian}(1991)
   for Problem 6.24, p.215 ), the inequalities (\ref{e31}) and (\ref{equivalent22})
yield
 \begin{eqnarray}\label{Appendx2}
\int_{0}^{t}d(e^{-\tilde{r}s}h(\tilde{S}_{s}))
&=&\int_{0}^{t}e^{-\tilde{r}s}\tilde{S}_{s}h^{'}(\tilde{S}_{s})\sigma
d\mathcal{W}(s)\nonumber\\
&&+\int_{0}^{t}e^{-\tilde{r}s}(\tilde{r}q-\delta \tilde{S}_{s}
) I_{\{L>\tilde{S}_{s}\geq b^*\}}ds\nonumber\\
&&+[h^{'}(L+)-h^{'}(L-)]\int_{0}^{
t}e^{-\tilde{r}s}dL_{s}(L)\nonumber\\
&\equiv & \mathcal{M}(t) + A_1(t) +\Lambda(t),
\end{eqnarray}
where $L_{t}(L)\geq0$  is  the local time of $\tilde{S}_{t} $ at the
point $ L$, $
\mathcal{M}(t)=\int_{0}^{t}e^{-\tilde{r}s}\tilde{S}_{s}h^{'}(\tilde{S}_{s})\sigma
d\mathcal{W}(s)$ is a  martingale,
\begin{eqnarray*}
A_1(t)&=&\int_{0}^{t}e^{-\tilde{r}s}(\tilde{r}q-\delta \tilde{S}_{s}
) I_{\{L>\tilde{S}_{s}\geq b^*\}}ds \qquad \mbox{and}\\
\Lambda(t)&=&  \ \ [h^{'}(L+)-h^{'}(L-)]\int_{0}^{
t}e^{-\tilde{r}s}dL_{s}(L).
\end{eqnarray*}
Define $T_n=t\wedge \tau \wedge \tau^{}_L -\frac{1}{n} $ for $
\forall t\geq 0 $ , $ n\geq 1 $   and  stopping time $ \tau $.  Then
by using definition of $\tau^{}_L $ and equality $L_t(L)=\int^t_0
I_{\{s: \tilde{S}_s=L\}}dL_s(L)$, we have $ \Lambda(T_n)=0$.
Moreover, since $ \tilde{r}\leq 0$ and $\delta\geq 0$, we see that
$A_1(T_n)\leq 0$. So by (\ref{Appendx2})
\begin{eqnarray*}\label{Appendx3}
h(\tilde{S}_{0})&=&{\bf E}\big [e^{-\tilde{r}T_n}h(\tilde{S}_{T_n})
I_{\{\tau^{}_{L}<\infty\}}\big ]+{\bf E} \big
[-A_1(T_n)I_{\{\tau^{}_{L}<\infty\}}\big ]
\nonumber\\
&\geq &{\bf E}\big [e^{-\tilde{r}T_n}h(\tilde{S}_{T_n})
I_{\{\tau^{}_{L}<\infty\}}\big ]\nonumber\\
&=&{\bf E}\big
[e^{-\tilde{r}T_n}h(\tilde{S}_{T_n})I_{\{\tau^{}_{L}<\infty\}}\big ]\nonumber\\
&\geq&{\bf E}\big [e^{-\tilde{r}T_n}(\tilde{S}_{T_n}\wedge
L-q)_{+}I_{\{\tau^{}_{L}<\infty\}}\big ].
\end{eqnarray*}
Note that
\begin{eqnarray*} e^{-\tilde{r}T_n}(\tilde{S}_{T_n}\wedge
L-q)_{+}I_{\{\tau^{}_{L}<\infty\}}\leq\sup\limits_{0\leq
t<\infty}e^{-\tilde{r}t}(\tilde{S}_{t}-q)_{+}
\end{eqnarray*}
and
 $${\bf E}\big [\sup\limits_{0\leq
t<\infty}e^{-\tilde{r}t}(\tilde{S}_{t}-q)_{+}\big]<\infty.$$
The dominated convergence theorem now implies
\begin{eqnarray*} \label{Appendx4}
h(x)=h(\tilde{S}_{0})&\geq&  \lim_{n\rightarrow +\infty}
{\bf E}\big [e^{-\tilde{r}T_n}(\tilde{S}_{T_n}\wedge
L-q)_{+}I_{\{\tau^{}_{L}<\infty\}}\big ]
\nonumber\\&=&{\bf E}\big
[e^{-\tilde{r}(\tau\wedge\tau^{}_{L})}(\tilde{S}_{\tau\wedge\tau^{}_{L}}\wedge
L-q)_{+}I_{\{\tau^{}_{L}<\infty\}} \big ].
\end{eqnarray*}
 Thus $h(x)=f(x)$ and
$\tau_{b}\wedge\tau^{}_{L}$ is the optimal stopping time.
\end{proof}
\vskip 5pt \noindent As a direct consequence of Theorem \ref{main2},
i.e.,  $ L=+\infty$,  we can get the following result proved by a
pure probability approach in Xia and Zhou\cite{stock}. \vskip 5pt
\noindent
\begin{CR}
Let $f(x)=\sup\limits_{\tau \in \mathcal {T}_{0}}{\bf E}\big [e^{-\tilde
{r}\tau}(\tilde{S}_{\tau}-q)_{+}\big ]$, $b$ and $ \lambda_1 $ be defined by
(\ref{b}) and (\ref{e35}) respectively.  Then
\begin{eqnarray}
f(x)=\left\{
\begin{array}{l l l}
(b^{}-q)(\frac{x}{b^{}})^{\lambda_{1}},&0< x\leq b^{},\\
x-q ,& x> b
\end{array}
\right.
\end{eqnarray}
and $ \tau_b  $ is the optimal stoping time.
\end{CR}
\vskip 5pt \noindent
\begin{Remark}Comparing with the value of capped American option with
non-negative interest rate(cf.Broadie and Detemple \cite{BRDE}(1995)
), we see that the value of capped stock loan treated in the present paper
 has  different behaviors.
 If the stock price is bigger than $L$ then the  capped American
option with non-negative interest rate should be exercise
immediately. While  the capped stock loan treated in the present paper
has no this kind of  performances.
\end{Remark}
\vskip 5pt \noindent
\begin{Remark} Let $V_t$ be the value process for the capped stock loan defined by (\ref{C24}). Then by using  Markov property of the process $S_{t}$
(cf. Karatzas and Shreve\cite{Methods}(1998)), we easily get
\begin{eqnarray*}
e^{-\tilde{r}t}f(\tilde{S}_{t})&=&\sup\limits_{\tau \in \mathcal
{T}_{t}}{\bf
E}\big[e^{-\tilde{r}(\tau\wedge\tau_{L})}(\tilde{S}_{\tau\wedge\tau_{L}}\wedge
L-q)_{+}I_{\{\tau_{L}<\infty\}}|\mathcal{F}_{t}\big]\\
&=&e^{-rt}V_{t}
\end{eqnarray*} and $V_{t}=e^{\gamma t}f(e^{-\gamma t}S_{t})$.
\end{Remark}
\vskip 5pt \noindent
 \setcounter{equation}{0}
\section{{\small {\bf Ranges of fair values
of parameters  }}} \vskip 10pt \noindent In this  section we will
work out the ranges of  fair values of the
 parameters  $(q, \gamma,c, L)$
 of the capped stock loan treated in this paper based on
  Theorem \ref{main2} and equality $f(S_0)=S_0
-q+c$. In view of Proposition\ref {Prop1}, there are two cases to be
dealt with. The first is when $ L>b$ and the second case is when $
L\leq b$. We only consider the first case, the second case can be
treated similarly. We shall distinguish three subcases, i.e.,
$S_{0}\geq L$, $L>S_{0}\geq b$ and $0< S_{0}\leq b$. \vskip 5pt
\noindent {\bf  Case of $S_{0}\geq L$}.  By (\ref{solution21}) and
$f(S_0)=S_0 -q+c$, we have
$f(S_0)=(L-q)(\frac{S_{0}}{L})^{\lambda_{2}}= S_0 -q+c$. So the
parameters  $\gamma$, $q$, $c$ and $L $  must satisfy
$c=(L-q)(\frac{S_{0}}{L})^{\lambda_{2}} +q -S_0 \leq 0$. This means
that the bank wants to earn maximal profit
 because the initial stock price is very high. The bank
  gets the stock via paying the money $c$ to
the client,  the client gives the stock to the bank due to the money
$c$. Therefore both the client and the bank have  incentives to do
the  business. Actually, $\tau^{}_{L}$ is the optimal stopping time
for the client to terminate the capped stock loans. \vskip 5pt
\noindent {\bf Case of  $L>S_{0}\geq b $}. By (\ref{solution21})
 and $f(S_0)=S_0 -q+c$, we have
$f(S_0)=S_{0}-q=S_{0}-q+c$, so $c=0$, i.e.,  the client has no
incentive to do the transaction.  $\tau_{b^{}}=0$ in this case is
the optimal stopping time.
 \vskip 5pt \noindent
 {\bf Case of  $0< S_{0}\leq b $}.\quad In this case
both the client and the bank have  incentives to do the business.
The bank does because there is dividend payment and so does the
client because the initial stock price is neither very high nor too
low. By Theorem \ref{main2}, the initial value is $f(S_{0})$. In
this case the bank  charges an amount $c=f(S_{0})-S_{0}+q$ from the
client for providing  the service. So the fair values of the
parameters $\gamma,q$ and $c$ must satisfy
\begin{eqnarray*}\label{value_determine}
S_{0}-q+c= (b^{}-q)(\frac{x}{b^{}})^{\lambda_{1}}
\end{eqnarray*}
and the optimal  stopping time is $\tau_{b^{}}$. \vskip 10pt
\noindent
 \setcounter{equation}{0}
\section{{\small {\bf  Examples }}}
\vskip 10pt \noindent In this section we will give two examples of
capped stock loan as follows.
\begin{EX}\label{E1}
Let the risk free rate $r=0.05$, the loan rate $\gamma=0.07$, the
volatility $\sigma=0.15$, the dividend $\delta=0.01$, the principal
$q=100$ and the cap $L=240$. Then $b=147.8< L$. We compute the
initial value $f(x)$ of capped stock loan  as in the following
Figure 1. The graph obviously shows that the cap greatly impacts on
the initial value when the stock price is very large. The cap
reduces the value of the stock loan and the client can acquires more
liquidity.
\begin{figure}[H]
  \includegraphics[height=9cm]{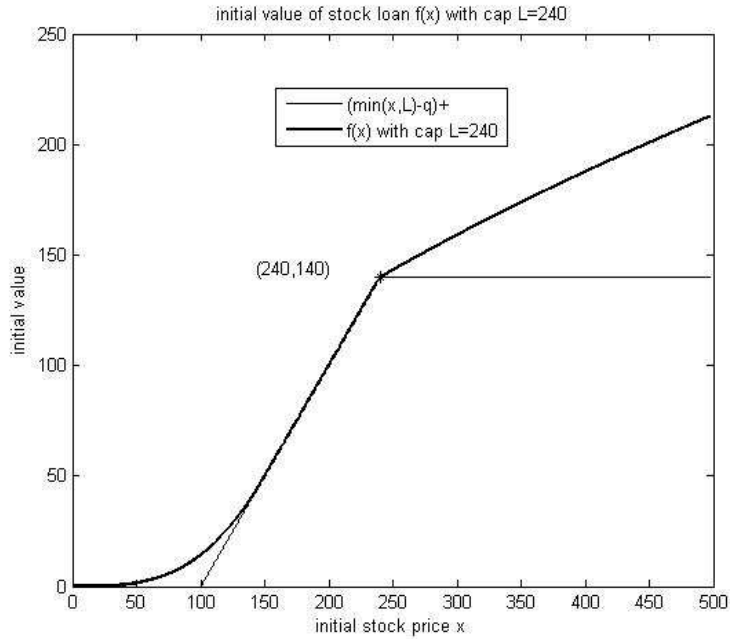}
  \caption{$\gamma=0.07,r=0.05,\sigma=0.15,\delta=0.01,
  q=100,cap=240$ }\label{fa}
\end{figure}
\end{EX}
\begin{EX}
Let $r$, $ \sigma$, $\sigma$, $\delta $ and $q$ be as the same as in
{\bf Example} \ref{E1},  $L=120$. Then $L < b=147.8$.  We compute
the initial value $f(x)$ of capped stock loan  as in the following
Figure 2. The graph explains  that  the initial value of caped stock
loan  greatly decreases  because the cap is less than $b^{}$. The
client will acquire much more liquidity compared to the uncapped
stock loan.
\begin{figure}[H]
  \includegraphics[height=9cm]{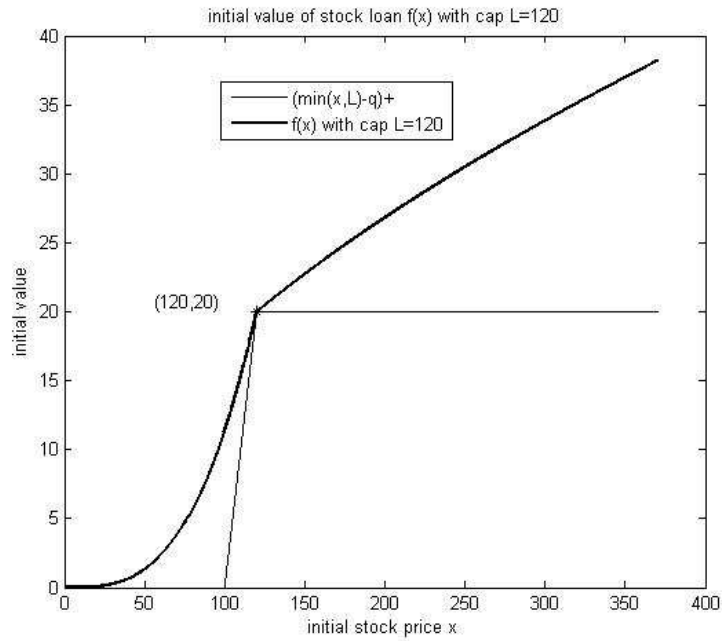}
  \caption{$\gamma=0.07,r=0.05,\sigma=0.15,\delta=0.01,q=100,cap=120$ }\label{fa}
\end{figure}
 \end{EX}
\vskip 15pt \noindent
 {\bf Acknowledgements.}  We
are very grateful to Professor Jianming Xia for his conversation
with us and providing  original paper of \cite{stock} for us. We
also express our deep thanks to Professor  Xun Yu Zhou for providing
power point files of his talk on stock loan at Peking University for
us. Special thanks also go to the participants of the seminar
stochastic analysis and finance  at Tsinghua University for their
feedbacks and useful conversations. This work is supported by
Project 10771114 of NSFC, Project 20060003001 of SRFDP, and SRF for
ROCS, SEM, and the Korea Foundation for Advanced Studies. We would
like to thank the institutions for the generous financial support.

\vskip 20pt \noindent \setcounter{equation}{0}


\begin{thebibliography}{99}
 \baselineskip17pt
\bibitem{Handbook} Borodin, A. N., Salminen, P., 2002.
 {\em Handbook of Brownian Motion-Facts and Formulae}.
  Probability and its Applications, Birkh$\ddot{a}$user Verlag, Basel, 2nd
  edition.
\bibitem{BRDE} Broadie, M., Detemple, J., 1995. {\em  American capped call
options on dividend-paying asset }. { The Review of Finance
Studies,} Vol. 8, No. 1, pp. 161-191.
\bibitem{Diffusion} Dayanik, S., Karatzas, I., 2003. {\em On the optimal stopping
problems for one-dimensional diffusions}.   Stochastic Process.
Appl., 107, no. 2, 173-212.
\bibitem{capped stock} Jiang, S., Liang, Z., Wu, W., 2008.  {\em Stock
loan with automatic termination clause}. Preprint(11, 2008).
 \bibitem{Brownian} Karatzas, I., Shreve,S. E., 1991.
  {\em Brownian motion and stochastic calculus}. Springer-Verlag, New York.
\bibitem{Methods} Karatzas, I.,  Shreve,S. E., 1998.
 {\em Methods of mathematical finance}. Springer-Verlag, New York.
\bibitem{Game} Kifer,Y., 2000. {\em Game Options}.
{ Finance and Stochastics}, 4:443-463.
\bibitem{Calculation} Kyprianou, A.E., 2004. {\em Some calculations for Israeli
options}.  Finance and Stochastics,  8:73-86.
\bibitem{s23}   Lions, P.-L., Sznitman, A.S.. 1984.  Stochastic differential
                 equations with reflecting boundary conditions. Comm.Pure Appl.
                 Math.{\bf 37}, 511-537.
\bibitem{McKean} McKean,H.P.JR., 1965. {\em A free-boundary problem for
the heat equation arising from a problem in mathematical economics}.
Industr. Manag. Rev., {\bf 6}, 32-39. Appendix to Samuelson(1965a).
\bibitem{Bernt and Sulem}  $\varnothing$ksendal,B. and Sulem,A.,
 2005. {\em Applied Stochastic Control of Jump Diffusions}. Springer.
\bibitem{Towards} Shiryaev,A.N.,   Kabanov,Y.M., Kramkov, D.O. and Melnikov,
A.V., 1994. {\em Towards the theory of options of both European and
American types II}. { Theory Prob.Appl}, 39,61-102.
\bibitem{stock} Xia, J.M.,  Zhou,X.Y.,
 2007.  {\em Stock loans}. { Mathematical Finance}, Vol.17, No.2, 307-317.
\end{thebibliography}
\end{document}